\documentclass[aps,pra,reprint,onecolumn,nofootinbib]{revtex4-1}
\usepackage{amsfonts,amsmath,amsthm}
\usepackage{MnSymbol}
\usepackage{xspace}

\newcommand{\QKD}{QKD\xspace}
\newcommand{\Xbasis}{{\mathtt X}\xspace}
\newcommand{\Zbasis}{{\mathtt Z}\xspace}
\newcommand{\Bbasis}{{\mathtt B}\xspace}
\newcommand{\FK}{FRKL\xspace}

\newtheorem{Def}{Definition}
\newtheorem{Thrm}{Theorem}

\newtheorem{Cor}{Corollary}
\newtheorem{Rem}{Remark}

\DeclareMathOperator{\ess}{ess}
\DeclareMathOperator{\Width}{Width}
\DeclareMathOperator{\BigOh}{O}

\begin{document}
\title{Application Of An Improved Version Of McDiarmid Inequality In
 Finite-Key-Length Decoy-State Quantum Key Distribution}

\author{H. F. Chau}
\thanks{email: \texttt{hfchau@hku.hk}}
\author{K.~C. Joseph Ng}
\affiliation{Department of Physics, University of Hong Kong, Pokfulam Road,
 Hong Kong}
\date{\today}

\begin{abstract}
 In practical decoy-state quantum key distribution, the raw key length is
 finite.
 Thus, deviation of the estimated single photon yield and single photon error
 rate from their respective true values due to finite sample size can
 seriously lower the provably secure key rate $R$.
 Current method to obtain a lower bound of $R$ follows an indirect path by
 first bounding the yields and error rates both conditioned on the type of
 decoy used.
 These bounds are then used to deduce the single photon yield and error rate,
 which in turn are used to calculate a lower bound of the key rate $R$.
 Here we report an improved version of McDiarmid inequality in statistics and
 show how use it to directly compute a lower bound of $R$ via the so-called 
 centering sequence.
 A novelty in this work is the optimization of the bound through the freedom
 of choosing possible centering sequences.
 The provably secure key rate of realistic 100~km long quantum channel
 obtained by our method is at least twice that of the state-of-the-art
 procedure when the raw key length $\ell_\text{raw}$ is $\approx 10^5$ to
 $10^6$.
 In fact, our method can improve the key rate significantly over a wide range
 of raw key length from about $10^5$ to $10^{11}$.
 More importantly, it is achieved by pure theoretical analysis without
 altering the experimental setup or the post-processing method.
 In a boarder context, this work introduces powerful concentration inequality
 techniques in statistics to tackle physics problem beyond straightforward
 statistical data analysis especially when the data are correlated so that
 tools like the central limit theorem are not applicable.
\end{abstract}

\maketitle

\section{Introduction}
\label{Sec:Intro}
 Quantum key distribution (\QKD) enables two trusted parties Alice and Bob to
 share a provably secure secret key by preparing and measuring quantum states
 that are transmitted through a noisy channel controlled by an eavesdropper
 Eve.
 One of the major challenges to make \QKD practical is to increase the number
 of secure bits generated per second~\cite{DLQY16}.
 That is why most \QKD experiments to date use photons as the quantum
 information carriers; and these photons come from phase randomize Poissonian
 distributed sources instead of the much less efficient single photon sources.
 In addition, decoy state method is used to combat Eve's
 photon-number-splitting attack on multiple photon events emitted from the
 Poissonian sources~\cite{Wang05,LMC05}.
 From the theoretical point of view, a more convenient figure of merit is the
 key rate, namely, the number of provably secure secret bits per average
 number of photon pulses prepared by Alice.
 This is because key rate measures the intrinsic performance of a \QKD
 protocol (in other words, the software issue) without taking the frequency of
 the pulse (which is a hardware issue) into account.
 This is analogous to the use of time complexity measure rather than the
 actual runtime to gauge the performance of an algorithm in theoretical
 computer science.

 Surely, provably secure lower bound of key rate $R$ (which we simply call the
 key rate from now on) of a \QKD scheme depends on various photon yields as
 well as error rates of those detected photons to be precisely defined in
 Eqs.~\eqref{E:Q_mu_def} and~\eqref{E:E_mu_def} below.
 The problem is that Alice and Bob can only transmit a finite number of
 photons in practice.
 Consequently, the yield and error rates estimated by any sampling technique
 may differ from their actual values.
 If Alice and Bob ignore these deviations, the actual number of bits of secret
 key they get could be smaller than that computed by the key rate $R$, posing
 a security threat.

 Various key rate formulae which take the above finite-size statistical
 fluctuations into account for a few (decoy-state-based) \QKD schemes had been
 reported in literature.
 For instance, Lim \emph{et al.}~\cite{LCWXZ14} computed the key rates of a
 certain implementation of the BB84 \QKD scheme~\cite{BB84} using three types
 of decoy; recently, Chau~\cite{Chau18} extended it to the case of using more
 than three types of decoys.
 Hayashi and Nakayama investigated the key rate for the BB84
 scheme~\cite{HN14}.
 Br\'{a}dler \emph{et al.} showed the key rate for a qudit-based \QKD
 scheme using up to three mutually unbiased preparation and measurement
 bases~\cite{BMFBB16}.
 And Wang \emph{et al.} proved that errors and fluctuations in the decoy
 photon intensities only have minor errors on the final key
 rate~\cite{XPZYP08}.
 In brief, the provably secure key rate of a \QKD scheme so far is found using
 the following three-step strategy.
 First, the yields $Q_{\Bbasis,\mu_n}$ and error rates $E_{\Bbasis,\mu_n}$
 conditioned on the preparation and measurement basis $\Bbasis$ as well as the
 photon intensity parameter $\mu_n$ used are determined by comparing the
 relevant Bob's measurement outcomes, if any, with Alice's preparation states.
 The second step is to deduce yields and error rates conditioned on the number
 of photons emitted by the source.
 For a phase randomized Poissonian photon source,
\begin{equation}
 Q_{\Bbasis,\mu_n} = \sum_{m=0}^{+\infty} \frac{\mu_n^m Y_{\Bbasis,m}
 \exp(-\mu_n)}{m!}
 \label{E:Q_mu_def}
\end{equation}
 and
\begin{equation}
 Q_{\Bbasis,\mu_n} E_{\Bbasis,\mu_n} = \sum_{m=0}^{+\infty} \frac{\mu_n^m
 Y_{\Bbasis,m} e_{\Bbasis,m} \exp(-\mu_n)}{m!} .
 \label{E:E_mu_def}
\end{equation}
 Here, $\mu_1 > \mu_2 > \cdots > \mu_k \ge 0$ are the photon intensities used
 in the decoy method with $k \ge 2$.
 Moreover, $Y_{\Bbasis,m}$ is the probability of photon detection by Bob given
 that the photon pulse sent by Alice contains $m$~photons and $e_{\Bbasis,m}$
 is the bit error rate for $m$~photon emission events prepared in the
 $\Bbasis$ basis~\cite{Wang05,LMC05,MQZL05}.
 The key rate $R$ depends on $Y_{\Bbasis,0}$, $Y_{\Bbasis,1}$ and
 $e_{\Bbasis,1}$~\cite{LCWXZ14,Wang05,LMC05,MQZL05}.
 Nevertheless, the later quantities cannot be determined precisely because
 Eqs.~\eqref{E:Q_mu_def} and~\eqref{E:E_mu_def} are under-determined systems
 of equations given $Q_{\Bbasis,\mu_n}$'s and $E_{\Bbasis,\mu_n}$'s provided
 that the number of photon intensities used $k$ is finite.
 To make things worse, in the finite-raw-key-length (\FK) situation, the
 measured values of $Q_{\Bbasis,\mu_n}$'s and $E_{\Bbasis,\mu_n}$'s deviate
 from their true values due to finite sampling.
 Fortunately, effective lower bounds of $Y_{\Bbasis,0}$ and $Y_{\Bbasis,1}$
 as well as upper bound of $e_{\Bbasis,1}$ are
 available~\cite{LCWXZ14,Wang05,LMC05,MQZL05,Hayashi07,Chau18}.
 In the \FK situation, these bounds can be deduced with the help of
 Hoeffding's inequality~\cite{Hoeffding}.
 (See, for example, Refs.~\cite{LCWXZ14,Chau18} for details.  Note that here
 we cannot assume the measurement outcomes are statistically independent and
 thus use more familiar tools such as central limit theorem because Eve may
 launch a coherent attack to all the photon pulses.  In fact, we do not even
 know what kind of statistical distributions do $Q_{\Bbasis,\mu_n}$'s and
 $E_{\Bbasis,\mu_n}$'s follow.)
 The third step is to deduce $R$ from these
 bounds~\cite{LCWXZ14,Wang05,LMC05,MQZL05,BMFBB16}.

 Computing lower bound of $R$ using this indirect strategy is not satisfactory
 in the \FK situation because it is unlikely for each of the finite-size
 fluctuations in $Q_{\Bbasis,\mu_n}$'s and $E_{\Bbasis,\mu_n}$'s to decrease
 the value of the provably secure key rate.
 In fact, for a given security parameter, the worst case bounds on
 $Y_{\Bbasis,0}$ and $Y_{\Bbasis,1}$ cannot be not attained simultaneously if
 the raw key length is finite.
 (This is evident, say, from the bounds of $Y_{\Bbasis,0}$ and $Y_{\Bbasis,1}$
 given by Inequalities~(2) and~(3) in Ref.~\cite{LCWXZ14} or
 Inequalities~(12a) and~(12b) in Ref.~\cite{Chau18}.  Note that there is a
 typo in Inequality~(12b) --- the $Q_{\Bbasis,\mu_i}^{\llangle k_0-i\rrangle}$
 there should be $Q_{\Bbasis,\mu_i}^{\llangle k_0-i+1\rrangle}$.  In all
 cases, the finite-size statistical fluctuation that leads to the saturation
 of lower bound for $Y_{\Bbasis,0}$ does not cause the saturation of the lower
 bound for $Y_{\Bbasis,1}$ and vice versa.)

 It is more effective if one could directly investigate the influence of
 finite-key-length on the key rate.
 To do so, one has to go beyond the use of Hoeffding's inequality to bound the
 statistical fluctuation, which only works for equally weighted sum of
 random variables that are either statistical independent or drawn from a
 finite population without replacement~\cite{Hoeffding}.
 Here we use the computation of the key rate of a specific BB84 \QKD
 protocol~\cite{BB84} that generates the raw key solely from $\Xbasis$ basis
 measurement results as an example to illustrate how to directly tackle
 statistical fluctuation in the \FK situation by means of McDiarmid-type
 inequality~\cite{McDiarmid} in statistics.
 The technique used here can be easily adapted to compute the key rates of
 other \QKD schemes using finite-dimensional qudits in the \FK situation.
 Our work here is based on an earlier preprint by one of the
 us~\cite{earlier}.
 Here we greatly extend and improve the original proposal by first proving a
 new and slightly extended McDiarmid-type of inequality on so-called
 centering sequences.
 (See Definition~\ref{Def:centering_property} for the precise definition of a
 centering sequence.)
 Then we apply it through four different methods, each giving a separate
 provably secure key rate.
 We also optimize the provably secure key rate $R$ by exploiting our freedom to
 pick the centering sequences.
 To our knowledge, this is the first time such an optimization is performed.
 In contrast, this type of optimization is not possible in previous approach
 that makes use of a less general inequality known as Hoeffding's inequality.
 It turns out that each method works best in different situations; and the
 best provably secure key rate among the four methods in realistic practical
 situation is at least about 10\% better than the state-of-the-art method
 before Ref.~\cite{earlier}.
 Moreover, for raw key length $\ell_\text{raw} \approx 10^5$ to $10^7$, this
 work almost double the secure key rate of the original proposal in
 Ref.~\cite{earlier} when four different photon intensities are used.
 From a broader perspective, the technique we introduce here is also
 applicable to bound the conclusion of a general physics experiment in the
 form of a real number due to finite-size statistical fluctuations of more
 than one type of measurement outcomes that are possibly statistically
 dependent.

\section{The \QKD Scheme By Chau In Ref.~\cite{Chau18} And The Assumptions Of
 The Security Proof}
\label{Sec:Scheme}
 To illustrate how McDiarmid-type of inequality can be used to give a better
 key rate, we consider the \QKD Scheme studied by Chau whose details can be
 found in Ref.~\cite{Chau18}.
 Note that this scheme is a slight variation of the one studied by Lim
 \emph{et al.} in Ref.~\cite{LCWXZ14}.
 The only difference is that they use three different photon intensities
 while we consider the slightly more general case of using $k\ge 2$ different
 photon intensities.
 In essence, the Scheme in Ref.~\cite{Chau18} is a decoy-state BB84 scheme
 with one-way classical communication using the $\Xbasis$-basis measurement
 results as the raw key and the $\Zbasis$-basis measurement results for phase
 error estimation.

 We assume that the light source is Poissonian distributed with intensities
 $\mu_1 > \mu_2 > \dots > \mu_k \ge 0$ with $k \ge 2$.
 Using the result in Ref.~\cite{XPZYP08}, we simply our discussion by assuming
 that these photon intensities are accurately determined and fixed throughout
 the experiment.
 This is fine because fluctuation of photon intensity of a laser source is
 negligible in practice.
 Since our aim is to demonstrate our technique of using McDiarmid-type
 inequality in the simplest possible \QKD implementation, we do not consider
 twin-field~\cite{LYDS18} or measurement device independent~\cite{LCQ12}
 setups although adaptation to these situations is straightforward though
 tedious.
 The measurement is performed using threshold photon detectors with random
 bit assignment in the event of multiple detector click.
 Last but not least, we assume both Alice and Bob have access to their own
 private perfect random number generators when choosing their preparation and
 measurement bases.

\section{Finite-Size Decoy-State Key Rate}
\label{Sec:Rate}
 Recall that the error rate for this particular variation of the decoy-state
 BB84 \QKD scheme using one-way classical communication is lower-bounded
 by~\cite{LCWXZ14,Chau18}
\begin{equation}
 p_{\Xbasis}^2 \left\{ \langle \exp(-\mu) \rangle Y_{\Xbasis,0} + \langle \mu
 \exp(-\mu) \rangle Y_{\Xbasis,1} [1-H_2(e_p)] - \Lambda_\text{EC} -
 \frac{\langle Q_{\Xbasis,\mu} \rangle}{\ell_\text{raw}} \left[ 6\log_2
 \frac{\chi}{\epsilon_\text{sec}} + \log_2 \frac{2}{\epsilon_\text{cor}}
 \right] \right\} ,
 \label{E:key_rate_basic}
\end{equation}
 where $p_\Xbasis$ denotes the probability that Alice (Bob) uses $\Xbasis$ as
 the preparation (measurement) basis, $\langle f(\mu)\rangle \equiv
 \sum_{n=1}^k p_{\mu_n} f(\mu_n)$ with $p_{\mu_n}$ being the probability for
 Alice to use photon intensity parameter $\mu_n$.
 Furthermore, $H_2(x) \equiv -x \log_2 x - (1-x) \log_2 (1-x)$ is the binary
 entropy function, $e_p$ is the phase error rate of the single photon events
 in the raw key, and $\Lambda_\text{EC}$ is the actual number of bits of
 information that leaks to Eve as Alice and Bob perform error correction on
 their raw bits.
 It is given by
\begin{equation}
 \Lambda_\text{EC} = \langle Q_{\Xbasis,\mu} H_2(E_{\Xbasis,\mu}) \rangle
 \label{E:information_leakage}
\end{equation}
 if they use the most efficient (classical) error correcting code to do the
 job.
 In addition, $\ell_\text{raw}$ is the raw sifted key
 length measured in bits, $\epsilon_\text{cor}$ is the upper bound of the
 chance that the final secret keys shared between Alice and Bob are different,
 and $\epsilon_\text{sec} = (1- p_\text{abort}) \| \rho_\text{AE} - U_\text{A}
 \otimes \rho_\text{E} \|_1 / 2$.
 Here $p_\text{abort}$ is the chance that the scheme aborts without generating
 a key, $\rho_\text{AE}$ is the classical-quantum state describing the joint
 state of Alice and Eve, $U_\text{A}$ is the uniform mixture of all the
 possible raw keys created by Alice, $\rho_\text{E}$ is the reduced density
 matrix of Eve, and $\| \cdot \|_1$ is the trace
 norm~\cite{Renner05,KGR05,RGK05}.
 Thus, Eve's information on the final key is at most $\epsilon_\text{sec}$.
 Last but not least, $\chi$ is a \QKD scheme specific factor which depends on
 the detailed security analysis used.
 In general, $\chi$ may also depend on other factors used in the \QKD scheme
 such as the number of photon intensities $k$~\cite{LCWXZ14,Chau18}.

 For BB84, $e_p \to e_{\Zbasis,1}$ as $\ell_\text{raw} \to +\infty$.
 More importantly, the best known bound on the difference between $e_p$ and
 $e_{\Zbasis,1}$ due to finite sample size correction using properties of
 the hypergeometric distribution reported in given by~\cite{Chau18,FMC10}
\begin{equation}
 e_p \le e_{\Zbasis,1} + \bar{\gamma} \left(
 \frac{\epsilon_\text{sec}}{\chi}, e_{\Zbasis,1}, \frac{s_\Zbasis
 Y_{\Zbasis,1} \langle \mu \exp(-\mu) \rangle}{\langle Q_{\Zbasis,\mu}
 \rangle}, \frac{s_\Xbasis Y_{\Xbasis,1} \langle \mu \exp(-\mu)
 \rangle}{\langle Q_{\Xbasis,\mu} \rangle} \right)
 \label{E:e_p_bound}
\end{equation}
 with probability at least $1-\epsilon_\text{sec}/\chi$, where
\begin{equation}
 \bar{\gamma}(a,b,c,d) \equiv \sqrt{\frac{(c+d)(1-b)b}{c d} \ \ln \left[
  \frac{c+d}{2\pi c d (1-b)b a^2} \right]} ,
 \label{E:gamma_def}
\end{equation}
 and $s_\Bbasis$ is the number of bits that are prepared and measured in
 $\Bbasis$ basis.
 Clearly, $s_\Xbasis = \ell_\text{raw}$ and $s_\Zbasis \approx (1-
 p_\Xbasis)^2 s_\Xbasis \langle Q_{\Zbasis,\mu} \rangle / (p_\Xbasis^2
 \langle Q_{\Xbasis,\mu} \rangle)$.
 (Note that $\bar{\gamma}$ becomes complex if $a,c,d$ are too large.
 This is because in this case no $e_p \ge e_{\Zbasis,1}$ exists with failure
 probability $a$.
 We carefully picked parameters here so that $\bar{\gamma}$ is real.)

 In the infinite-key-length limit, statistical fluctuations of
 $Q_{\Bbasis,\mu_n}$ and $E_{\Bbasis,\mu_n}$ can be ignored.
 Then based on the analysis in Ref.~\cite{Chau18} with typos corrected, one
 has
\begin{subequations}
\label{E:parameter_bounds}
\begin{equation}
 Y_{\Bbasis,0} \ge \max \left( 0, \sum_{n=1}^k a_{0n} Q_{\Bbasis,\mu_n}
 \right) \equiv \max \left( 0, \sum_{n=k_0}^k \frac{-Q_{\Bbasis,\mu_n}
 \exp[\mu_n] \hat{\prod}_{i\ne n} \mu_i}{\hat{\prod}_{j\ne n} [\mu_n - \mu_j]}
 \right) ,
 \label{E:Y0_bound}
\end{equation}
\begin{equation}
 Y_{\Bbasis,1} \ge \max \left( 0, \sum_{n=1}^k a_{1n} Q_{\Bbasis,\mu_n}
 \right) \equiv \max \left( 0, \sum_{n=3-k_0}^k \frac{-Q_{\Bbasis,\mu_n}
 \exp[\mu_n] \hat{S}_n}{\hat{\prod}_{j\ne n} [\mu_n - \mu_j]} \right)
 \label{E:Y1_bound}
\end{equation}
 and
\begin{equation}
 Y_{\Zbasis,1} e_{\Zbasis,1} \le \min \left( \frac{Y_{\Zbasis,1}}{2} ,
 \sum_{n=1}^k a_{2n} Q_{\Zbasis,\mu_n} E_{\Zbasis,\mu_n} \right) \equiv \min
 \left( \frac{Y_{\Zbasis,1}}{2} , \sum_{n=k_0}^k \frac{Q_{\Zbasis,\mu_n}
 E_{\Zbasis,\mu_n} \exp[\mu_n] \hat{S}_n}{\hat{\prod}_{j\ne n} [\mu_n -
 \mu_j]} \right) ,
 \label{E:Ye1_bound}
\end{equation}
 where $k_0 = 1 (2)$ if $k$ is even (odd), and $\hat{\prod}_{j\ne n}$ is over
 the dummy variable $j$ from $k_0$ to $k$ but skipping $n$.
 In addition, $\hat{S}_n = \sum'' \mu_{t_1} \mu_{t_2} \cdots
 \mu_{t_{k-k_0-1}}$ where the double primed sum is over $k_0 \le t_1 < t_2 <
 \cdots < t_{k-k_0-1} \le k$ with $t_1,t_2,\dots,t_{k-k_0-1} \ne n$.
 (In other words, $a_{01} = a_{21} = 0$ if $k$ is odd and $a_{11} = 0$ if $k$
 is even.)
 Note that in our subsequent analysis, we also need the following two
 inequalities, which can be proven using the same method as in
 Inequality~\eqref{E:Y1_bound}:
\begin{equation}
 Y_{\Zbasis,1} \bar{e}_{\Zbasis,1} \equiv Y_{\Zbasis,1} (1 - e_{\Zbasis,1})
 \ge \max \left( 0, \sum_{n=1}^k a_{1n} Q_{\Zbasis,\mu_n}
 \bar{E}_{\Zbasis,\mu_n} \right) \equiv \max \left( 0, \sum_{n=3-k_0}^k
 \frac{-Q_{\Zbasis,\mu_n} \bar{E}_{\Zbasis,\mu_n} \exp[\mu_n]
 \hat{S}_n}{\hat{\prod}_{j\ne n} [\mu_n - \mu_j]} \right)
 \label{E:Y1bare1_bound}
\end{equation}
 and
\begin{equation}
 Y_{\Zbasis,1} e_{\Zbasis,1} \ge \max \left( 0, \sum_{n=1}^k a_{1n}
 Q_{\Zbasis,\mu_n} E_{\Zbasis,\mu_n} \right) \equiv \max \left( 0,
 \sum_{n=3-k_0}^k \frac{-Q_{\Zbasis,\mu_n} E_{\Zbasis,\mu_n} \exp [\mu_n]
 \hat{S}_n}{\prod_{j\ne n} [\mu_n - \mu_j]} \right) ,
 \label{E:Ye1_special_bound}
\end{equation}
\end{subequations}
 where $\bar{E}_{\Zbasis,\mu_n} = 1 - E_{\Zbasis,\mu_n}$.

 Substituting Inequalities~\eqref{E:e_p_bound}
 and~\eqref{E:parameter_bounds} into Expression~\eqref{E:key_rate_basic} gives
 the following lower bound of the key rate
\begin{equation}
 \sum_{n=1}^k b_n Q_{\Xbasis,\mu_n} - p_\Xbasis^2 \left\{ \Lambda_\text{EC} +
 \frac{\langle Q_{\Xbasis,\mu} \rangle}{\ell_\text{raw}} \left[ 6\log_2
 \frac{\chi}{\epsilon_\text{sec}} + \log_2 \frac{2}{\epsilon_\text{cor}}
 \right] \right\} ,
 \label{E:key_rate}
\end{equation}
 where
\begin{equation}
 b_n = p_\Xbasis^2 \left\{ \langle \exp(-\mu) \rangle a_{0n} + \langle \mu
 \exp(-\mu) \rangle a_{1n} [1-H_2(e_p)] \right\}
 \label{E:b_n_def}
\end{equation}
 provided that $Y_{\Xbasis,0}, Y_{\Xbasis,1} > 0$.
 (The cases of $Y_{\Xbasis,0}$ or $Y_{\Xbasis,1} = 0$ can be dealt with in
 the same way by changing the definition of $b_n$ accordingly.  But these
 cases are not interesting for they likely imply $R = 0$ in realistic
 channels.)
  
 Note that the worst case key rate corresponds to the situation that the spin
 flip and phase shift errors in the raw key are uncorrelated so that Alice and
 Bob cannot use the correlation information to increase the efficiency of
 entanglement distillation.
 Thus, we may separately consider statistical fluctuations in
 $Q_{\Xbasis,\mu_n}$'s and $e_{\Zbasis,1}$ in the \FK situation.

\section{An Improved Version Of McDiarmid Inequality}
\label{Sec:McDiarmid}
 We now prove an improved version of a deep mathematical statistics result
 before applying it to improve the key rate $R$.
 Our inslight is that statistical fluctuations in $Q_{\Xbasis,\mu_n}$'s and
 $e_{\Zbasis,1}$ can be bounded using McDiarmid-type inequality.
 Actually, the first inequality of this type was proven for the case of
 statistically independent random variables using martingale technique in
 Ref.~\cite{McDiarmid}.
 The inequality we need here is a straightforward extension of Theorem~6.7 in
 Ref.~\cite{McDiarmid} and Theorem~2.3 in Ref.~\cite{McDiarmid1} for
 statistically dependent random variables.
 (See also a closely related version in Ref.~\cite{McDiarmid2}.)

 We first introduce the concept of a centering sequence~\cite{McDiarmid1}.
 The definition below is written in a more apparent manner to physicists.

\begin{Def}
 Let ${\mathbf W} = (W_1,W_2,\ldots,W_t)$ be a random real vector whose
 components $W_i$'s are possibly statistically dependent random variables each
 taking values in the set ${\mathcal W}_i$.
 Let $f_m$ be a real-valued bounded function of ${\mathbf W}$.
 Set $V_m = \left. f_m({\mathbf W}) \right|_{B_m}$ where $B_m$ denotes the
 conditions $W_j = w_j$ for $j = 1,2,\ldots,m-1$.
 Then, the sequence of random variable $\{ V_m \}_{m=1}^t$ is said to be
 {\bf centering} if $E[ U_m \mid V_{m-1} = v] \equiv E[ V_m - V_{m-1} \mid
 V_{m-1} = v]$ is a decreasing functions of $v$ for all $m = 1,2,\ldots,t$.
 (Here we use the convention that $V_0 = 0$ and assume that all conditional
 expectation values $E[\cdot \mid \cdot]$ exist.)
 \label{Def:centering_property}
\end{Def}

 Note that centering property implicitly depends on the distribution of
 ${\mathbf W}$ through the conditional expectation value of $U_m$.
 Moreover, $\{ V_m \}$ is centering if $\{ V_m \}$ is a martingale.

\begin{Thrm}
 \label{Thrm:McDiarmid}
 Using notations in Definition~\ref{Def:centering_property}, for a fixed
 $i=1,2,\ldots,t$, let $w_m \in {\mathcal W}_m$ and set
 \begin{align}
  \hat{r}_m(w_1,\ldots,w_{m-1}) &= \ess\sup \{ E[ U_m({\mathbf W}) \mid W_m =
   w_m] \}_{w_m\in {\mathcal W}_m} - \ess\inf \{ E[ U_m({\mathbf W}) \mid W_m
   = w'_m] \}_{w'_m \in {\mathcal W}_m} \nonumber \\
  &\equiv b_m(w_1,\ldots,w_{m-1}) - a_m(w_1,\ldots,w_{m-1}) .
  \label{E:range_def}
 \end{align}
 Here the symbols $\ess\sup$ and $\ess\inf$ denote the essential supremum
 and infimum, respectively.
 Further set $\hat{r}^2 \equiv \hat{r}^2(w_1,\ldots,w_{t-1}) = \sum_{m=1}^t
 \hat{r}_m^2$.
 Then $f_t({\mathbf w}) \equiv f_t(w_1,w_2,\ldots,w_t)$ obeys
 \begin{subequations}
 \label{E:McDiarmid}
 \begin{equation}
  \Pr(f_t({\mathbf w}) - E[f_t({\mathbf W})] \ge \delta) \le \exp \left[
  \frac{-2\delta^2}{\hat{r}^2(w_1,\ldots,w_{t-1})} \right]
  \label{E:McDiarmid1}
 \end{equation}
 and
 \begin{equation}
  \Pr(f_t({\mathbf w}) - E[f_t({\mathbf W})] \le -\delta) \le \exp \left[
  \frac{-2\delta^2}{\hat{r}^2(w_1,\ldots,w_{t-1})} \right]
  \label{E:McDiarmid2}
 \end{equation}
 \end{subequations}
 for any $\delta > 0$, where $\Pr(\cdot)$ denotes the occurrence probability
 of the argument.
\end{Thrm}

\begin{Rem}
 This version of McDiarmid inequality is slightly stronger than the one
 reported in Ref.~\cite{McDiarmid} as we also utilize information of
 ${\mathbf w}$ in obtaining $\hat{r}$ whereas the original version in
 Ref.~\cite{McDiarmid} made use of the worst case ${\mathbf w}$.
 The proof of this theorem is based on that of Theorem~2.2 in
 Ref.~\cite{McDiarmid1}.
 \label{Rem:McDiarmid_difference}
\end{Rem}

\begin{proof}
 Note that for any $h,\delta>0$,
 \begin{alignat}{2}
 & \Pr(V_t - E[V_t] \ge \delta) \nonumber \\
 \le{}& E[\exp \{h (V_t - E[V_t] - \delta)\}] =
  e^{-h(\delta + E[V_t])} E[\exp ( h V_t)] &&
  \quad \text{(by Bernstein's inequality)} \nonumber \\
 ={}& e^{-h(\delta + E[V_t])} E[\exp( h V_{t-1} ) E[ \exp( h U_t) \mid
  V_{t-1}]] \nonumber \\
 \le{}& e^{-h(\delta + E[V_t])} E \left[ \exp( h V_{t-1} ) \left\{ \frac{(b_t
  - E[U_t \mid V_{t-1}]) e^{h a_t}}{b_t - a_t} \right. \right. \nonumber \\
 & \qquad \left. \left. + \frac{(E[U_t \mid V_{t-1}] - a_t) e^{h b_t}}{b_t -
  a_t} \right\} \right] && \quad
   \begin{matrix}
    \text{(since~} a_t \le E[U_t \mid V_{t-1}] \le b_t
     \text{~and the line joining} \hfill \\
    \,(a_t,e^{h a_t}) \text{~and~} (b_t,e^{h b_t}) \text{~is above the curve~}
     y = e^{h x} \hfill \\
    \,\text{for~} x\in [a_t,b_t]) \hfill
   \end{matrix} \nonumber \\
 \le{}& e^{-h(\delta + E[V_t])} E[\exp (h V_{t-1})] \left\{ \frac{(b_t - E[U_t
  \mid V_{t-1}]) e^{h a_t}}{b_t - a_t} \right. \nonumber \\
 & \qquad \left. + \frac{(E[U_t \mid V_{t-1}] - a_t) e^{h b_t}}{b_t - a_t}
  \right\} && \quad
  \text{(by Chebyshev's sum inequality on centering sequence)} \nonumber \\
 \le{}& e^{-h(\delta + E[V_{t-1}])} e^{-h E[U_t \mid V_{t-1}]} \left\{
  \frac{(b_t - E[U_t \mid V_{t-1}]) e^{h a_t}}{b_t - a_t} \right. \nonumber \\
 & \qquad \left. + \frac{(E[U_t \mid V_{t-1}] - a_t) e^{h b_t}}{b_t - a_t}
  \right\} . && \quad \text{(by Jensen's inequality)}
 \label{E:main_proof_1}
\end{alignat}

 To proceed, we consider the function $g(h) = -h x + \ln \{ [(b_t - x)
 e^{h a_t} + (x - a_t) e^{h b_t}] / (b_t - a_t) \}$ for $x\in [a_t,b_t]$.
 It is straightforward to check that $g(0) = \left. dg/dh \right|_{h=0} = 0$.
 Moreover,
\begin{equation}
 \frac{d^2 g}{d h^2} = \frac{(b_t-a_t)^2 (b-x)(x-a) e^{h(b_t+a_t)}}{\left[
 (b-x) e^{h a_t} + (x-a) e^{h b_t}\right]^2} \le \frac{(b_t-a_t)^2}{4}
 \label{E:dg2dh2} 
\end{equation}
 with the equality holds whenever $(b-x) e^{h a_t} = (x-a) e^{h b_t}$.
 Therefore, Taylor's theorem gives $g(h) \le h^2 (b_t - a_t)^2 / 8$ for all
 $h\ge 0$.
 Applying this inequality with $x = E[U_t \mid V_{t-1}]$ to
 Inequality~\eqref{E:main_proof_1}, we have
\begin{align}
 \Pr(f_t({\mathbf w}) - E[f_t({\mathbf W})] \ge \delta) = \Pr(V_t - E[V_t]
  \ge \delta) &\le e^{-h(\delta+E[V_{t-1}])} e^{h^2(b_t-a_t)^2/8} \nonumber \\
 &\le \exp \left[ -\delta h + \frac{h^2 \sum_{m=1}^t (b_m - a_m)^2}{8}
  \right] = \exp \left( \frac{h^2 \hat{r}^2}{8} - \delta h \right)
 \label{E:main_proof_2}
\end{align}
 for any $h>0$.
 The R.H.S. of Inequality~\eqref{E:main_proof_2} is minimized by setting
 $h = 4\delta / \hat{r}^2$; and with this $h$,
 Inequality~\eqref{E:main_proof_2} becomes Inequality~\eqref{E:McDiarmid1}.

 Finally, by applying the same argument to $-f_m$'s instead of $f_m$'s, we
 get Inequality~\eqref{E:McDiarmid2}.
 This completes our proof.
\end{proof}

\begin{Cor}
 Let ${\mathbf W} = (W_1,\ldots,W_t)$ be a random vector such that $W_m$ takes
 on value from the same bounded set of real numbers ${\mathcal W} = \{
 \alpha_j \}_{j=1}^k$ for all $m = 1,2,\ldots,t$.
 Suppose further that $W_m$'s are multivariate hypergeometrically distributed.
 Let $f_m({\mathbf W}) = \sum_{i=1}^m W_i$ for all $m = 1,2,\ldots,t$.
 Then, the sequence of random variables $\{ V_m \}_{m=1}^t$ defined in
 Definition~\ref{Def:centering_property} is centering provided that $E[V_m -
 V_{m-1} \mid V_{m-1} = v]$ is well-defined for all $v$.
 Besides, Theorem~\ref{Thrm:McDiarmid} holds with $\hat{r} = \sqrt{t}
 \Width({\mathcal W})$, where $\Width({\mathcal W}) \equiv \ess\sup
 {\mathcal W} - \ess\inf {\mathcal W}$.
 \label{Cor:centering_sum}
\end{Cor}
\begin{proof}
 This proof is adapted from Example~1 in Ref.~\cite{McDiarmid1}.
 From Definition~\ref{Def:centering_property}, it suffices to show that
 $E[U_m \mid \sum_{i=1}^{m-1} W_i = \upsilon] = E[ W_m \mid \sum_{i=1}^{m-1}
 W_i = \upsilon]$ is a decreasing function of $\upsilon$.
 Suppose $W_i$'s are drawn from a collection of $M$ objects out of which
 $M_j$ of them take the value $\alpha_j$ for all $j$.
 Suppose further that among $W_i$'s with $1\le i < m$, there are $m_j$ of
 them taking the value of $\alpha_j$ for all $j$.
 Then, the probability that $W_m = \alpha_j$ is $(M_j-m_j)/(M-m+1)$.
 Moreover, the condition $\sum_{i=1}^{m-1} W_i = \upsilon$ means that $\sum
 m_j \alpha_j = \upsilon$.
 As a result, $E[W_m \mid \sum_{i=1}^{m-1} W_i = \upsilon] = \sum_j (M_j-m_j)
 \alpha_j/(M-m+1) = (\sum_j M_j \alpha_j - \upsilon)/(M-m+1)$, which is a
 decreasing function of $\upsilon$ whenever $\sum_{i=1}^{m-1} W_i = \upsilon$.
 Hence, $\{ V_m \}$ is a centering sequence.

 By applying Theorem~\ref{Thrm:McDiarmid} to $\{ U_m \}$, we have $r_m = \ess
 \sup \{ W_m \mid V_{m-1} = \upsilon \} - \ess\inf \{ W_m \mid V_{m-1} =
 \upsilon \} = \ess\sup {\mathcal W} - \ess\inf {\mathcal W} = \Width
 ({\mathcal W})$ for all $m$ and $V_{m-1}$.
 Hence, it is proved.
\end{proof}

\begin{Rem}
 The above corollary was first proven by Hoeffding in Ref.~\cite{Hoeffding}
 without using the concept of centering sequence.
 Actually, Corollary~\ref{Cor:centering_sum} is more often referred to as the
 Hoeffding's inequality.
 In fact, Hoeffding's inequality has been used to compute the provably secure
 key rate $R$ when the raw key length $\ell_\text{raw}$ is finite in previous
 works~\cite{LCWXZ14,HN14,BMFBB16,Chau18}.
 In Sec.~\ref{Sec:Appl} below, we use the above corollary to bound
 $e_{\Zbasis,1}$ in Methods~\ref{Method:conventional} and~\ref{Method:Ybare}.
 \label{Rem:relation_to_Hoeffding}
\end{Rem}

\begin{Cor}
 Let ${\mathbf W} = (W_1,\ldots,W_t)$ be a random vector where each $W_m$
 takes on value from a bounded set of real numbers ${\mathcal W} = \{ \alpha_j
 \}_{j=1}^k$.
 Suppose $W_m$'s are multivariate hypergeometrically distributed in the sense
 that they are chosen without replacement from a collection of $M$ objects
 out of which $M_j$ of them take the value of $\alpha_j$ for all $j$.
 Let $x\in [\ess\inf {\mathcal W},\ess\sup {\mathcal W}]$ and $y>0$ be two
 fixed numbers.
 Let $P\colon \{1,2,\ldots, t\} \mapsto \{1,2,\ldots,t\}$ be an arbitrary but
 fixed permutation.
 Suppose
 \begin{equation}
  y + \sum_{i=1}^t W_{P(i)} > \ess\sup {\mathcal W} - \ess\inf {\mathcal W}
  \ge 0 .
  \label{E:centering_f_condition}
 \end{equation}
 Define
 \begin{equation}
  f_m({\mathbf W}) = \frac{(t-m) x + \sum_{i=1}^m W_{P(i)}}{y + (t-m) x +
  \sum_{i=1}^m W_{P(i)}}
  \label{E:centering_f_def}
 \end{equation}
 for $m=1,2,\ldots,t$.
 Then, the sequence $\{ V_m \}_{m=1}^t$ is centering provided that
 \begin{equation}
  x \le \min_{m=1}^t \frac{2\sum_{j=1}^k M_j \alpha_j - \sup \sum_{i=1}^{m-1}
  W_{P(i)} + y - \delta}{2M-t-m+1}
  \label{E:centering_f_condition2}
 \end{equation}
 where $\delta$ is a small correlation term of the order of $\Width
 ({\mathcal W})/(y+t x)^2$.
 Furthermore, by picking $x$ to be the R.H.S. of
 Inequality~\eqref{E:centering_f_condition2}, then
 Inequality~\eqref{E:McDiarmid} is true with
 \begin{equation}
  \hat{r}^2 = \sum_{m=1}^t \left\{ \frac{y \Width({\mathcal W})}{[y+(t-m)x+
  \ess\sup {\mathcal W}+ \sum_{i=1}^{m-1} w_{P(i)}][y+(t-m)x+\ess\inf
  {\mathcal W}+\sum_{i=1}^{m-1} w_{P(i)}]} \right\}^2 ,
  \label{E:centering_f_r}
 \end{equation}
 where $\{ w_{P(i)} \}$ is a decreasing sequence.
 \label{Cor:centering_f}
\end{Cor}

\begin{proof}
 Since $(W_{P(1)},W_{P(2)},\ldots,W_{P(t)})$ is also a multivariate
 hypergeometrically distributed random vector, we only need to prove the case
 when $P$ is an identity operator as the general case can be proven in the
 same way.
 From Eq.~\eqref{E:centering_f_condition}, $f_m$ has a positive denominator
 and is an increasing function of $W_m$.
 So to prove that $\{ V_m \}$ is centering, it suffices to show that $E[U_m
 \mid \sum_{i=1}^{m-1} W_i = (m-1)w]$ is a decreasing function of $(m-1)w =
 \sum_{j=1}^k m_j \alpha_j$ for all non-negative integers $m_j$'s obeying
 $\sum_{j=1}^k m_j = m-1$.
 Since $\alpha_j$'s are fixed, the only way to change $w$ is to change $m_j$'s
 but at the same time keeping $\sum_{j=1}^n m_j$ fixed.
 Clearly, $w$ can only be changed if $m\ge 3$.
 More importantly, as $m_j$'s are integers, any such change can be expressed
 as a composition of a series of elementary changes, each increases a certain
 $m_{j_1}$ by one and decreasing a certain $m_{j_2}$ by one with $1\le j_1 \ne
 j_2 \le k$.

 Observe that
\begin{align}
 & E \left[ U_m \mid \sum_{i=1}^{m-1} W_i = (m-1) w = \sum_{j=1}^k m_j
  \alpha_j \right] \nonumber \\
 ={}& E \left[ \frac{y (W_m - x)}{\{ y + (t-m+1)x + (m-1)w \} \{ y + (t-m)x +
  (m-1)w + W_m\}} \mid \sum_{i=1}^{m-1} W_i = (m-1)w = \sum_{j=1}^k m_j
  \alpha_j \right] \nonumber \\
 ={}& \frac{y}{(M-m+1) [y+(t-m+1)x+(m-1)w]} \sum_{j=1}^k \frac{(M_j-m_j)
  (\alpha_j - x)}{y+(t-m)x+(m-1)w+\alpha_j} \nonumber \\
 \equiv{}& \frac{y}{(M-m+1) D} \sum_{j=1}^k \frac{(M_j-m_j) (\alpha_j - x)}{D
  + \alpha_j - x} .
 \label{E:centering_f_formal_sum}
\end{align}
 Moreover, after the elementary change, $w \mapsto w + \alpha_{j_1} -
 \alpha_{j_2} \equiv w + \Delta w$.
 From Inequality~\eqref{E:centering_f_condition}, $D > |\alpha_j-x|$.
 So by Taylor's theorem,
\begin{align}
 & E \left[ U_m \mid \sum_{i=1}^{m-1} W_i = (m-1) w = \sum_{j=1}^k m_j
  \alpha_j \right] \nonumber \\
 \longmapsto{}& E \left[ U_m \mid \sum_{i=1}^{m-1} W_i = (m-1) w + \Delta w =
  \alpha_{j_1} - \alpha_{j_2} + \sum_{j=1}^k m_j \alpha_j \right] \nonumber \\
 ={}& \frac{y}{(M-m+1) (D+\Delta w)} \left\{ \sum_{j=1}^k \frac{(M_j - m_j)
  (\alpha_j-x)}{D+\Delta w} \left[ 1 - \frac{\alpha_j-x}{D+\Delta w} + \xi_1
  \left( \frac{\alpha_j-x}{D+\Delta w} \right)^2 \right] +
  \frac{-(\alpha_{j_1} - x)}{D+\Delta w+\alpha_{j_1}} + \frac{\alpha_{j_2} -
  x}{D+\Delta w+\alpha_{j_2}} \right\} \nonumber \\
 ={}& \frac{y}{(M-m+1) (D+\Delta w)^2} \left[ \sum_{j=1}^k M_j \alpha_j -
  (M-m+1) x - (m-1)w - \frac{\sum_{j=1}^n (M_j - m_j) (\alpha_j - x)^2}{D+
  \Delta w} + \xi_1 \left( \frac{\alpha_j-x}{D+\Delta w} \right)^2 \right]
  \nonumber \\
 & \qquad - \frac{y\Delta w}{(M-m+1)(D+\Delta w+\alpha_{j_1}-x) (D+\Delta w+
  \alpha_{j_2}-x)}
 \label{E:centering_f_formal_sum_change}
\end{align}
 with $\xi_1\in [0,1]$.
 As $x\in [\ess\inf {\mathcal W},\ess\sup {\mathcal W}]$, we conclude that $0
 \le \sum_j (M_j - m_j)(\alpha_j - x)^2 \le (M-m+1) \Width({\mathcal W})^2$
 almost surely.
 From Inequality~\eqref{E:centering_f_condition}, we may expand
 $1/(D+\Delta w)$, $1/(D+\Delta w+\alpha_{j_1}-x)$ and $1/(D+\Delta w+
 \alpha_{j_2}-x)$ as series of $\Delta w$ via Taylor's theorem.
 In this way, the R.H.S. of Eq.~\eqref{E:centering_f_formal_sum_change} can be
 expressed in the form $E[ U_m \mid \sum_{i=1}^{m-1} W_i = (m-1) w =
 \sum_{j=1}^k m_j \alpha_j ] + g_1 \Delta w + g_2 $ with
\begin{align}
 g_1 &= -\frac{y}{(M-m+1) D^2} \left\{ \frac{2 \left[ \sum_{j=1}^k M_j
  \alpha_j - (M-m+1)x - (m-1)w \right]}{D} + \left[ 1 - \frac{\xi_2
  (\alpha_{j_1}-x)}{D} \right] \left[ 1 - \frac{\xi_3 (\alpha_{j_2}-x)}{D}
  \right] \right. \nonumber \\
 & \qquad \left. \vphantom{\frac{\left[\sum_{j=1}^n M_j\alpha_j\right]}{D}}
  - \frac{3 \sum_{j=1}^k (M_j - m_j)(\alpha_j - x)^2}{D^2} \right\} \nonumber
  \\
 &\le{} -\frac{y}{(M-m+1) D^2} \left\{ \frac{2 \left[ \sum_{j=1}^k M_j
  \alpha_j - (M-m+1)x - (m-1)w \right]}{D} + \left[ 1 - \frac{\Width
  ({\mathcal W})}{D} \right]^2 - \frac{3(M-m+1) \Width({\mathcal W})^2}{D^2}
  \right\} ,
 \label{E:centering_f_linear_coef}
\end{align}
 where $\xi_2,\xi_3\in [0,1]$.
 And the correlation term $g_2$ obeys $|g_2| \le 3y [\sum_j M_j \alpha_j -
 (M-m+1)x - (m-1)w] (\Delta w)^2 /[(M-m+1) D^4]$.

 A sufficient condition for $\{ V_m \}$ to be centering is $g_1 \Delta w + g_2
 \le 0$ for all $m$ and $w$.
 Moreover, this condition is satisfied if
\begin{equation}
 x \le \frac{2\sum_{j=1}^k M_j \alpha_j - (m-1)w + y - \delta}{2M-t-m+1}
 \label{E:centering_f_condition_inter}
\end{equation}
 for all $m = 1,\ldots,t$ and for all $(m-1) w = \sum_{i=1}^{m-1} W_i$, where
 the correlation term $\delta \le |g_2|\Delta w + 2 \Width({\mathcal W}) +
 (3M-3m+2) \Width({\mathcal W})^2/D$.
 (Note that Inequality~\eqref{E:centering_f_condition_inter} is consistent
 with the constraint that $\ess\inf {\mathcal W} \le x \le \ess\sup
 {\mathcal W}$ because this inequality is trivially satisfied when $x =
 \ess\inf {\mathcal W}$.)
 Hence, $\{ V_m \}$ is centering if
 Inequality~\eqref{E:centering_f_condition2} holds.

 We now switch back to consider the situation of an arbitrary but fixed
 permutation $P$.
 To optimize the bound in Theorem~\ref{Thrm:McDiarmid}, we use the freedom to
 pick a suitable permutation $P$ to minimize $\hat{r}$.
 From Theorem~\ref{Thrm:McDiarmid}, $\hat{r}_m = y \Width({\mathcal W}) / \{[y
 + (t-m) x + \ess\sup {\mathcal W} + \sum_{i=1}^{m-1} w_{P(i)}] [y + (t-m) x +
 \ess\inf {\mathcal W} + \sum_{i=1}^{m-1} w_{P(i)}] \}$, which is a decreasing
 function of both $x$ and $w$.
 Hence, the optimal situation occurs when we pick the permutation so that
 $w_{P(i)}$ is a decreasing function of $i$.
 In this case, $\sum_{i=1}^m w_{P(i)} / m$ is a decreasing function of $m$.
 In this way, we arrive at $\hat{r}^2$ in Eq.~\eqref{E:centering_f_r}.
\end{proof}

\begin{Rem}
 The ability to optimize $\hat{r}$ by means of picking the best possible
 permutation $P$ and hence the best possible centering sequence is a novel
 feature of McDiarmid inequality.
 As far as we know, this feature has not been exploited before.
 In contrast, from the proof of Corollary~\ref{Cor:centering_sum}, it is clear
 that the value of $\hat{r}$ obtained from the Hoeffding's inequality does not
 depend on the choice of $P$.
 In Sec.~\ref{Sec:Appl} below, we fully exploit this freedom of picking $P$ to
 bound $e_{\Zbasis,1}$ in Method~\ref{Method:aggressive}.
 Note however that the above Corollary requires the knowledge of $M_j$'s.
 In addition, $\hat{r}^2$ is written as a rather involved sum.
 Let us replace every $w_{P(i)}$ in Eq.~\eqref{E:centering_f_r} by the average
 observed value, namely, $\sum_{i=1}^t w_i / t$.
 In this way, $\hat{r}$ would increase by a factor of $\BigOh(t \Width
 ({\mathcal W})/D)$.
 Suppose that we fix $x = \sum_{i=1}^t w_i / t \equiv \langle w \rangle$ as
 well (without caring whether Inequality~\eqref{E:centering_f_condition2}
 holds or not).
 Then $\hat{r}$ would change by a factor of $\BigOh(\Width({\mathcal W})/D
 \sqrt{t})$
 most of the time due to statistical fluctuation.
 Thus, in practice, we may replace $\hat{r}$ in Eq.~\eqref{E:centering_f_r} by
 the following more convenient and useful expression
 \begin{equation}
  \hat{r} = \frac{\sqrt{t} y \Width({\mathcal W})}{[y+(t-1)\langle w\rangle+
  \ess\inf {\mathcal W}] [y+(t-1)\langle w\rangle+\ess\sup {\mathcal W}]} ,
  \label{E:centering_f_hat_r_approx}
 \end{equation}
 which does not depend on the knowledge of $M_j$'s.
 This expression for $\hat{r}$ shall be used to bound $e_{\Zbasis,1}$ in
 Method~\ref{Method:direct} to be reported in Sec.~\ref{Sec:Appl}.
 \label{Rem:approximate_solution}
\end{Rem}

\section{Application Of The Improved McDiarmid Inequality In Finding The Key
 Rate}
\label{Sec:Appl}
 There is a subtlety in applying Theorem~\ref{Thrm:McDiarmid} to study the
 statistical fluctuation of $e_{\Zbasis,1}$.
 A naive way to do so is to use Inequalities~\eqref{E:e_p_bound}
 and~\eqref{E:parameter_bounds} to obtain the bound $e_{\Zbasis,1} \le
 (\sum_{n=1}^k a_{2n} Q_{\Zbasis,\mu_n} E_{\Zbasis,\mu_n}) / (\sum_{n=1}^k
 a_{1n} Q_{\Zbasis,\mu_n})$.
 Then one could regard $Q_{\Zbasis,\mu_n}$'s and $Q_{\Zbasis,\mu_n}
 E_{\Zbasis,\mu_n}$'s as random variables and directly apply
 Theorem~\ref{Thrm:McDiarmid} and Definition~\ref{Def:centering_property} to
 the R.H.S. of the above inequality.
 Nonetheless, it does not work for the R.H.S. of this inequality need not be
 bounded.
 Besides, the bound obtained is not strong enough even if we ignore the
 boundedness problem.

 To proceed, we first write $Q_{\Zbasis,\mu_n} = \sum_j \tilde{W}_{nj} /
 \tilde{s}_{\Zbasis,\mu_n}$ where $\tilde{s}_{\Zbasis,\mu_n}$ is the number of
 photon pulses that Alice prepares using photon intensity $\mu_n$ and that
 Alice prepares and Bob tries to measure (but may or may not have detection)
 in $\Zbasis$ basis.
 In addition, $\tilde{W}_{nj}$ denotes the possibly correlated random variable
 whose value is $1$ ($0$) if the $j$th photon pulse among the
 $\tilde{s}_{\Zbasis,\mu_n}$ photon pulses is (not) detected by Bob.
 Clearly, $\tilde{s}_{\Zbasis,\mu_n} \approx T p_\Zbasis^2 p_{\mu_n}$ with $T$
 being the total number of photon pulses sent by Alice and $p_\Zbasis = 1 -
 p_\Xbasis$ is the probability for Alice (Bob) to prepare (measure) in the
 $\Zbasis$ basis.
 Since $s_\Zbasis \approx T p_\Zbasis^2 \langle Q_{\Zbasis,\mu} \rangle$, I
 arrive at
\begin{subequations}
\label{E:Q_Zbasis_random_variable_expressions}
\begin{equation}
 Y_{\Zbasis,1} \ge \max \left( 0, \sum_{n=1}^k a_{1n} Q_{\Zbasis,\mu_n}
 \right) = \max \left( 0, \frac{\langle Q_{\Zbasis,\mu} \rangle}{s_\Zbasis}
 \sum_{n=1}^k \left\{ \frac{a_{1n}}{p_{\mu_n}} \left[ \sum_j \tilde{W}_{nj}
 \right] \right\} \right) = \max \left( 0, \frac{\langle Q_{\Zbasis,\mu}
 \rangle}{s_\Zbasis} \sum_{i=1}^{s_\Zbasis} W_{\Zbasis,i} \right) .
 \label{E:Q_Zbasis_random_variables_expression_for_Y}
\end{equation}
 Here $W_{\Zbasis,i}$ is the random variable that takes the value $a_{1n} /
 p_{\mu_n}$ if the $i$th photon pulse that are prepared by Alice and then
 successfully measured by Bob both in the $\Zbasis$ basis is in fact prepared
 using photon intensity $\mu_n$.
 Recall that Eve knows the number of photons in each pulse and may act
 accordingly.
 However, she does not know the photon intensity parameter used in each pulse
 and the preparation basis until the pulse is measured by Bob.
 Hence, $W_{\Zbasis,n}$'s may be correlated.
 Actually, the most general situation is that $W_{\Zbasis,n}$'s are drawn from
 a larger population without replacement.
 That is to say, these random variables obey the multivariate hypergeometric
 distribution.
 By the same argument, Inequalities~\eqref{E:Ye1_bound}
 and~\eqref{E:Y1bare1_bound} gives
\begin{equation}
 Y_{\Zbasis,1} e_{\Zbasis,1} \le \min \left( \frac{Y_{\Zbasis,1}}{2},
 \sum_{n=1}^k a_{2n} Q_{\Zbasis,\mu_n} E_{\Zbasis,\mu_n} \right) = \min \left(
 \frac{Y_{\Zbasis,1}}{2}, \frac{\langle Q_{\Zbasis,\mu} \rangle}{s_{\Zbasis}}
 \sum_{i=1}^{s^\text{e}_{\Zbasis}} W^\text{e}_{\Zbasis,i} \right)
 \label{E:Q_Zbasis_random_variables_expression_for_Ye}
\end{equation}
 and
\begin{equation}
 Y_{\Zbasis,1} \bar{e}_{\Zbasis,1} \ge \max \left( 0, \sum_{n=1}^k a_{1n}
 Q_{\Zbasis,\mu_n} \bar{E}_{\Zbasis,\mu_n} \right) = \max \left( 0,
 \frac{\langle Q_{\Zbasis,\mu} \rangle}{s_{\Zbasis}}
 \sum_{i=1}^{s^{\bar{e}}_{\Zbasis}} W^{\bar{e}}_{\Zbasis,i} \right) ,
 \label{E:Q_Zbasis_random_variables_expression_for_Ybare}
\end{equation}
\end{subequations}
 where $s^\text{e}_{\Zbasis} = s_{\Zbasis} \langle Q_{\Zbasis,\mu}
 E_{\Zbasis,\mu} \rangle / \langle Q_{\Zbasis,\mu} \rangle$ and
 $s^{\bar{\text{e}}}_{\Zbasis} = s_{\Zbasis} \langle Q_{\Zbasis,\mu}
 \bar{E}_{\Zbasis,\mu} \rangle / \langle Q_{\Zbasis,\mu} \rangle$) are the
 number of bits that are prepared and successfully measured in the $\Zbasis$
 basis such that the preparation by Alice and measurement result by Bob are
 unequal and equal, respectively.
 Moreover, $W^\text{e}_{\Zbasis,i}$'s ($W^{\bar{\text{e}}}_{\Zbasis,i}$'s) are
 multivariate hypergeometrically distributed random variables taking values in
 the set $\{ a_{2n} / p_{\mu_n} \}_{n=1}^k$ ($\{ a_{1n} / p_{\mu_n}
 \}_{n=1}^k$).

 From Inequalities~\eqref{E:Q_Zbasis_random_variables_expression_for_Y}
 --\eqref{E:Q_Zbasis_random_variables_expression_for_Ybare}, $e_{\Zbasis,1}$
 obeys
\begin{subequations}
\label{E:e_Z_1_inequalities}
\begin{equation}
 e_{\Zbasis,1} \le \max \left( 0, \min \left( \frac{1}{2},
 \frac{\sum_{i=1}^{s^\text{e}_{\Zbasis}}
 W^\text{e}_{\Zbasis,i}}{\sum_{j=1}^{s_{\Zbasis}} W_{\Zbasis,j}} \right)
 \right)
 \label{E:e_Z_1_expression1}
\end{equation}
 and
\begin{equation}
 e_{\Zbasis,1} \le \max \left( 0 , \min \left( \frac{1}{2},
 \frac{\sum_{i=1}^{s^\text{e}_{\Zbasis}}
 W^\text{e}_{\Zbasis,i}}{\sum_{i=1}^{s^\text{e}_{\Zbasis}}
 W^\text{e}_{\Zbasis,i} +
 \sum_{j=1}^{s^{\bar{\text{e}}}_{\Zbasis}} W^{\bar{\text{e}}}_{\Zbasis,j}}
 \right) \right) .
 \label{E:e_Z_1_expression2}
\end{equation}
\end{subequations}

 Interestingly, these two inequalities can be used to give four different
 bounds on the finite-size statistical fluctuations in $e_{\Zbasis,1}$.
 More importantly, these four bounds are
\begin{enumerate}
 \item \label{Method:conventional}
  Use an upper bound of $\sum_{i=1}^{s_\Zbasis^\text{e}}
  W_{\Zbasis,i}^\text{e}$ and a lower bound of $\sum_{j=1}^{s_\Zbasis}
  W_{\Zbasis,j}$ to deduce an upper bound of $e_{\Zbasis,1}$.
  Specifically, from Corollary~\ref{Cor:centering_sum}, we conclude that the
  true value of $\sum_{j=1}^{s_\Zbasis} W_{\Zbasis,j}$ is less than the
  observed value by $\left[ s_\Zbasis \ln (1/\epsilon_\Zbasis)/2 \right]^{1/2}
  \Width(\{a_{1n} / p_{\mu_n} \}_{n=1}^k)$ with probability at most
  $\epsilon_\Zbasis$.
  (Recall that $\Width({\mathcal W})$ of a bounded set ${\mathcal W}$ of real
  numbers is defined as $\ess\sup {\mathcal W} - \ess\inf {\mathcal W}$.)
  And the true value of $\sum_{i=1}^{s_\Zbasis^\text{e}}
  W_{\Zbasis,i}^\text{e}$ is greater than its observed value by $\left[
  s^\text{e}_\Zbasis \ln (1/\epsilon^\text{e}_\Zbasis) / 2 \right]^{1/2}
  \Width(\{ a_{2n} / p_{\mu_n} \}_{n=1}^k) = \left[ s_\Zbasis \langle
  Q_{\Zbasis,\mu} E_{\Zbasis,\mu} \rangle \ln (1/\epsilon^\text{e}_\Zbasis) /
  2 \langle Q_{\Zbasis,\mu} \rangle \right]^{1/2} \Width (\{ a_{2n} /
  p_{\mu_n} \}_{n=1}^k)$ with probability at most $\epsilon_\Zbasis^\text{e}$.
  Since $W_{\Zbasis_i}^\text{e}$ and $W_{\Zbasis,j}$ are positively correlated,
  from Inequalities~\eqref{E:Q_Zbasis_random_variables_expression_for_Y},
  \eqref{E:Q_Zbasis_random_variables_expression_for_Ye}
  and~\eqref{E:e_Z_1_expression1}, we have
  \begin{subequations}
  \begin{equation}
   e_{\Zbasis,1} \le \max \left( 0, \min \left( \frac{1}{2},
   \frac{\sum_{n=1}^k a_{2n} Q_{\Zbasis,\mu_n} E_{\Zbasis,\mu_n} + \Delta
   Y_{\Zbasis,1} e_{\Zbasis,1}}{\sum_{n=1}^k a_{1n} Q_{\Zbasis,\mu_n} - \Delta
   Y_{\Zbasis,1}} \right) \right)
   \label{E:e_Z1_bound_method1}
  \end{equation}
  with probability at least $1-\epsilon_\Zbasis-\epsilon_\Zbasis^\text{e}$,
  where
  \begin{equation}
   \Delta Y_{\Zbasis,1} e_{\Zbasis,1} = \left[ \frac{\langle Q_{\Zbasis,\mu}
   \rangle \langle Q_{\Zbasis,\mu} E_{\Zbasis,\mu} \rangle \ln
   (1/\epsilon^\text{e}_\Zbasis)}{2 s_\Zbasis} \right]^{1/2} \Width \left(
   \left\{ \frac{a_{2n}}{p_{\mu_n}} \right\}_{n=1}^k \right)
   \label{E:Delta_Ye_def}
  \end{equation}
  and
  \begin{equation}
   \Delta Y_{\Zbasis,1} = \langle Q_{\Zbasis,\mu} \rangle \left[ \frac{\ln
   (1/\epsilon_\Zbasis)}{2 s_\Zbasis} \right]^{1/2} \Width \left( \left\{
   \frac{a_{1n}}{p_{\mu_n}} \right\}_{n=1}^k \right) .
   \label{E:Delta_Y_def}
  \end{equation}
  Incidentally, this is the method reported in the preprint by one of us in
  Ref.~\cite{earlier}.
  Moreover, similar bounds on statistical fluctuations of $Q_{\Bbasis,n}$'s
  and $Q_{\Bbasis,1} E_{\Bbasis,1}$ have been obtained using Hoeffding's
  inequality in Refs.~\cite{LCWXZ14,Chau18}.
  That method is not as effective as the one reported here since they
  indirectly deal with finite sampling statistical fluctuation of
  $Y_{\Zbasis,1}$ and $Y_{\Zbasis,1} e_{\Zbasis,1}$.
 \item \label{Method:Ybare}
  Alternatively, we may use Inequality~\eqref{E:e_Z_1_expression2} and
  Corollary~\ref{Cor:centering_sum} to bound $e_{\Zbasis,1}$.
  Specifically, the true value of $\sum_{j=1}^{s_\Zbasis^{\bar{\text{e}}}}
  W_{\Zbasis,j}^{\bar{\text{e}}}$ is less than the observed value by
  $[s_\Zbasis \langle Q_{\Zbasis,\mu} \bar{E}_{\Zbasis,\mu} \rangle \ln (1/
  \epsilon_\Zbasis^{\bar{\text{e}}}) / 2\langle Q_{\Zbasis,\mu} \rangle]^{1/2}
  \Width(\{ a_{1n}/p_{\mu_n} \}_{n=1}^k)$ with probability at most
  $\epsilon_\Zbasis^{\bar{\text{e}}}$.
  Note that $W_{\Zbasis,j}^\text{e}$'s and $W_{\Zbasis,j}^{\bar{\text{e}}}$'s
  are statistically independent.
  Therefore, from
  Inequalities~\eqref{E:Q_Zbasis_random_variables_expression_for_Ye},
  \eqref{E:Q_Zbasis_random_variables_expression_for_Ybare}
  and~\eqref{E:e_Z_1_expression2}, we have
  \begin{equation}
   e_{\Zbasis,1} \le \max \left( 0, \min \left( \frac{1}{2},
   \frac{\sum_{n=1}^k a_{2n} Q_{\Zbasis,\mu_n} E_{\Zbasis,\mu_n} + \Delta
   Y_{\Zbasis,1} e_{\Zbasis,1}}{\sum_{n=1}^k a_{1n} Q_{\Zbasis,\mu_n}
   \bar{E}_{\Zbasis,\mu_n} + \sum_{n=1}^k a_{2n} Q_{\Zbasis,\mu_n}
   E_{\Zbasis,\mu_n} - \Delta Y_{\Zbasis,1} \bar{e}_{\Zbasis,1} + \Delta
   Y_{\Zbasis,1} e_{\Zbasis,1}} \right) \right)
   \label{E:e_Z1_bound_method2}
  \end{equation}
  with probability at least $1 - \epsilon_\Zbasis^\text{e} -
  \epsilon_\Zbasis^{\bar{\text{e}}}$, where $\Delta Y_{\Zbasis,1}
  e_{\Zbasis,1}$ is given by Eq.~\eqref{E:Delta_Ye_def} and
  \begin{equation}
   \Delta Y_{\Zbasis,1} \bar{e}_{\Zbasis,1} = \left[ \frac{\langle
   Q_{\Zbasis,\mu} \rangle \langle Q_{\Zbasis,\mu} \bar{E}_{\Zbasis,\mu}
   \rangle \ln (1/\epsilon^\text{e}_\Zbasis)}{2 s_\Zbasis} \right]^{1/2}
   \Width \left( \left\{ \frac{a_{1n}}{p_{\mu_n}} \right\}_{n=1}^k \right) .
   \label{E:Delta_Ybare_def}
  \end{equation}
 \item \label{Method:direct}
  An even more interesting way to bound $e_{\Zbasis,1}$ is to use
  Inequality~\eqref{E:e_Z_1_expression2}, Corollary~\ref{Cor:centering_f} and
  Remark~\ref{Rem:approximate_solution}.
  Since $\langle w\rangle$ in this case is the measured $Y_{\Zbasis,1}
  e_{\Zbasis,1} / s_\Zbasis^\text{e}$, which is lower-bounded by
  Inequality~\eqref{E:Ye1_special_bound},
  Remark~\ref{Rem:approximate_solution} gives
  \begin{equation}
   e_{\Zbasis,1} \le \max \left( 0, \min \left( \frac{1}{2},
   \frac{\sum_{n=1}^k a_{2n} Q_{\Zbasis,\mu_n} E_{\Zbasis,\mu_n}}{\sum_{n=1}^k
   a_{1n} Q_{\Zbasis,\mu_n} \bar{E}_{\Zbasis,\mu_n} + \sum_{n=1}^k a_{2n}
   Q_{\Zbasis,\mu_n} E_{\Zbasis,\mu_n} - \Delta Y_{\Zbasis,1}
   \bar{e}_{\Zbasis,1}} + \Delta e_{\Zbasis,1} \right) \right)
   \label{E:e_Z1_bound_method3}
  \end{equation}
  with probability at least $1-\epsilon_\Zbasis^{\bar{\text{e}}}-
  \epsilon_\Zbasis^\text{e}$, where
  \begin{align}
   \Delta e_{\Zbasis,1} &= \left[ \frac{\langle Q_{\Zbasis,\mu} \rangle
    \langle Q_{\Zbasis,\mu} E_{\Zbasis,\mu} \rangle \ln
    (1/\epsilon_\Zbasis^\text{e})}{2 s_\Zbasis} \right]^{1/2} \left(
    \sum_{n=1}^k a_{1n} Q_{\Zbasis,\mu_n} \bar{E}_{\Zbasis,\mu_n} - \Delta
    Y_{\Zbasis,1} \bar{e}_{\Zbasis,1} \right) \Width \left( \left\{
    \frac{a_{2n}}{p_{\mu_n}} \right\}_{n=1}^k \right) \nonumber \\
   & \quad \times \left[ \sum_{n=1}^k a_{1n} Q_{\Zbasis,\mu_n}
    \bar{E}_{\Zbasis,\mu_n} - \Delta Y_{\Zbasis,1} \bar{e}_{\Zbasis,1} +
    \left( 1 - \frac{\langle Q_{\Zbasis,\mu} \rangle}{s_\Zbasis \langle
    Q_{\Zbasis,\mu} E_{\Zbasis,\mu} \rangle} \right) \sum_{n=1}^k a_{1n}
    Q_{\Zbasis,\mu_n} E_{\Zbasis,\mu_n} + \frac{\langle Q_{\Zbasis,\mu}
    \rangle^2}{s_\Zbasis^2 \langle Q_{\Zbasis,\mu} E_{\Zbasis,\mu} \rangle}
    \max_{n=1}^k \left\{ \frac{a_{2n}}{p_{\mu_n}} \right\} \right]^{-1}
    \nonumber \\
   & \quad \times \left[ \sum_{n=1}^k a_{1n} Q_{\Zbasis,\mu_n}
    \bar{E}_{\Zbasis,\mu_n} - \Delta Y_{\Zbasis,1} \bar{e}_{\Zbasis,1} +
    \left( 1 - \frac{\langle Q_{\Zbasis,\mu} \rangle}{s_\Zbasis \langle
    Q_{\Zbasis,\mu} E_{\Zbasis,\mu} \rangle} \right) \sum_{n=1}^k a_{1n}
    Q_{\Zbasis,\mu_n} E_{\Zbasis,\mu_n} + \frac{\langle Q_{\Zbasis,\mu}
    \rangle^2}{s_\Zbasis^2 \langle Q_{\Zbasis,\mu} E_{\Zbasis,\mu} \rangle}
    \min_{n=1}^k \left\{ \frac{a_{2n}}{p_{\mu_n}} \right\} \right]^{-1}
  \end{align}
  provided that Inequalities~\eqref{E:centering_f_condition}
  and~\eqref{E:centering_f_condition2} hold.
  (See Ref.~\cite{Joseph} for an alternative proof of this result.)
 \item \label{Method:aggressive}
  There is an alternative way to apply Inequality~\eqref{E:e_Z_1_expression2}
  and Corollary~\ref{Cor:centering_f} to find $\Delta e_{\Zbasis,1}$ in
  Inequality~\eqref{E:e_Z1_bound_method3}, which is quite aggressive.
  Since $\sum_{i=1}^t w_i / t$ is an estimate of $Y_{\Zbasis,1} e_{\Zbasis,1}
  = \langle w\rangle$, we know from Corollary~\ref{Cor:centering_sum} and
  Inequality~\eqref{E:Ye1_special_bound} that $\langle w\rangle \ge
  \sum_{n=1}^k a_{1n} Q_{\Zbasis,\mu_n} E_{\Zbasis,\mu_n} - \Delta
  Y_{\Zbasis,1} e_{\Zbasis,1}$ with probability at least $1-
  \epsilon_\Zbasis^\text{e}$.
  In other words, by fixing $t = s_\Zbasis^\text{e}$ and $x = (\sum_{n=1}^k
  a_{1n} Q_{\Zbasis,\mu_n} E_{\Zbasis,\mu_n} - \Delta Y_{\Zbasis,1}
  e_{\Zbasis,1})/s_\Zbasis^\text{e}$, we conclude that
  Inequality~\eqref{E:centering_f_condition_inter} is satisfied with
  probability at least $1-\epsilon_\Zbasis^\text{e}$.
  Next, we could upper-bound the R.H.S. of Eq.~\eqref{E:centering_f_r} by
  approximating the sum over $m$ there by an integral.
  Specifically, set $y = \sum_{n=1}^k a_{1n} Q_{\Zbasis,\mu_n}
  \bar{E}_{\Zbasis,\mu_n} - \Delta Y_{\Zbasis,1} \bar{e}_{\Zbasis,1}$, then
  \begin{align}
   \hat{r}^2 &\approx \sum_{m=1}^t \int_0^1 \left\{ \frac{y}{y+(t-m+1)x+
    \ess\inf {\mathcal W} +\sum_{i=1}^{m-1} w_{P(i)} + \mu [w_{P(m)}-x]}
    \right. \nonumber \\
   & \qquad \left. - \frac{y}{y+(t-m+1)x+\ess\sup {\mathcal W} +
    \sum_{i=1}^{m-1} w_{P(i)} + \mu [w_{P(m)}-x]} \right\}^2 d\mu \nonumber \\
   &= y^2 \sum_{m=1}^t \frac{1}{w_{P(m)}-x} \left( - \frac{1}{y+(t-m+1)x+
    \ess\inf {\mathcal W} + \sum_{i=1}^{m-1} w_{P(i)} + \mu[ w_{P(m)}-x]}
    \right. \nonumber \\
   & \qquad - \frac{1}{y+(t-m+1)x+\ess\sup {\mathcal W} + \sum_{i=1}^{m-1}
    w_{P(i)} + \mu[ w_{P(m)}-x]} \nonumber \\
   & \qquad \left. \left. + \frac{2}{\Width({\mathcal W})} \ln \left\{
    \frac{y+(t-m+1)x+\ess\sup {\mathcal W} + \sum_{i=1}^{m-1} w_{P(i)} + \mu
    [w_{P(m)}-x]}{y+(t-m+1)x+\ess\inf {\mathcal W} + \sum_{i=1}^{m-1} w_{P(i)}
    + \mu [w_{P(m)}-x]} \right\} \right) \right|_{\mu = 0}^1 .
    \label{E:method4_integral}
  \end{align}
  As $w_{P(m)}$'s are arranged in descending order and ${\mathcal W}$ in our
  case is the set $\{ \langle Q_{\Zbasis,\mu}\rangle a_{2n}/(p_{\mu_n}
  s_\Zbasis) \}_{n=1}^k$ of at most $k$ elements, R.H.S. of the above
  inequality can be simplified to a big sum of at most $k$ terms.
  To be more explicit, suppose the descending sequence $\{ w_{P(m)}
  \}_{m=1}^t$ contains $n^{(1)}$ copies of $w^{(1)}$, followed by $n^{(2)}$
  copies of $w^{(2)}$, and so on until ending with $n^{(k)}$ copies of
  $w^{(k)}$.
  Surely, $\sum_{i=1}^k n^{(i)} = t = s_\Zbasis^\text{e}$ and $\sum_{i=1}^k
  n^{(i)} w^{(i)} / t$ is the observed $\langle w\rangle$.
  Then, Eq.~\eqref{E:method4_integral} becomes
  \begin{align}
   \hat{r}^2 &\approx y^2 \sum_{m=1}^k \frac{1}{w^{(m)}-x} \left( -
    \frac{1}{y+(t-\sum_{i<m} n^{(i)}+1)x+ \ess\inf {\mathcal W} + \sum_{i<m}
    n^{(i)} w^{(i)} + \mu[ w^{(m)}-x]} \right. \nonumber \\
   & \qquad - \frac{1}{y+(t-\sum_{i<m} n^{(i)}+1)x+\ess\sup {\mathcal W} +
    \sum_{i<m} n^{(i)} w^{(i)} + \mu[ w^{(m)}-x]} \nonumber \\
   & \qquad \left. \left. + \frac{2}{\Width({\mathcal W})} \ln \left\{
    \frac{y+(t-\sum_{i<m} n^{(i)}+1)x+\ess\sup {\mathcal W} + \sum_{i<m}
    n^{(i)} w^{(i)} + \mu [w^{(m)}-x]}{y+(t-\sum_{i<m} n^{(i)}+1)x+\ess\inf
    {\mathcal W} + \sum_{i<m} n^{(i)} w^{(i)} + \mu [w^{(m)}-x]} \right\}
   \right) \right|_{\mu = 0}^{n^{(m)}} ,
   \label{E:method4_integral_explicit}
  \end{align}
  which is efficient to compute.
  (Note that the sum in R.H.S. of Eq.~\eqref{E:centering_f_r} is a decreasing
  function of $x$ and $w_{P(i)}$'s, the above integral approximation is
  accurate up to a correction term of at most $\Width({\mathcal W})^2 =
  \BigOh(1/s_\Zbasis^2)$.
  Surely, this correction can be safely ignored in practice provided that
  $t = s_\Zbasis^\text{e} \gtrsim 10^4$.)
  In this way, Inequality~\eqref{E:e_Z1_bound_method3} holds with probability
  at least $1-\epsilon_\Zbasis^{\bar{\text{e}}}-2\epsilon_\Zbasis^\text{e}$
  with
  \begin{equation}
   \Delta e_{\Zbasis,1} = \hat{r} \left[ \frac{\ln
   (1/\epsilon_\Zbasis^\text{e})}{2} \right]^{1/2} ,
   \label{E:e_Z1_bound_method4}
  \end{equation}
  \end{subequations}
  where $\hat{r}$ is given by the R.H.S. of
  Inequality~\eqref{E:method4_integral_explicit}.
\end{enumerate}

 In reality, we use the minimum of the above four methods to upper-bound the
 value of $e_{\Zbasis,1}$.
 To study the statistical fluctuation of $R$, it remains to consider the
 fluctuation of $Q_{\Xbasis,\mu_n}$ in the first term in
 Expression~\eqref{E:key_rate}.
 (Although the second term also depends on $Q_{\Xbasis,\mu_n}$'s implicitly
 through $\Lambda_\text{EC}$, statistical fluctuation is absent from this
 term.
 This is because $\Lambda_\text{EC}$ is the amount of information leaking to
 Eve during classical post-processing of the measured raw bits.
 Thus, it depends on the observed values of $Q_{\Xbasis,\mu_n}$'s and
 $E_{\Xbasis,\mu_n}$'s instead of their true values.)
 Using the same technique as in the estimation of statistical fluctuation in
 $e_{\Zbasis,1}$, the first term of Expression~\eqref{E:key_rate} can be
 rewritten as $\langle Q_{\Xbasis,\mu} \rangle \sum_{i=1}^{s_\Xbasis}
 W_{\Xbasis,i}$ where $W_{\Xbasis,i}$'s are multivariate hypergeometrically
 distributed random variables each taken values in the set $\{ b_n / p_{\mu_n}
 \}_{n=1}^k$.
 Here $b_n$ is given by Eq.~\eqref{E:b_n_def} with $e_p$ equals the R.H.S. of
 Inequality~\eqref{E:e_p_bound} where $e_{\Zbasis,1}$ is given by any one of
 the following four equations depending on which of the four methods we use:
 Eq.~\eqref{E:e_Z1_bound_method1}, \eqref{E:e_Z1_bound_method2},
 \eqref{E:e_Z1_bound_method3} and~\eqref{E:e_Z1_bound_method4}.
 Corollary~\ref{Cor:centering_sum} implies that due to statistical
 fluctuation, the true value of the first term in
 Expression~\eqref{E:key_rate} is lower than the observed value by
 $\langle Q_{\Xbasis,\mu} \rangle \left[ \ln (1/\epsilon_\Xbasis) / (2
 s_\Xbasis) \right]^{1/2} \Width (\{ b_n / p_{\mu_n} \}_{n=1}^k)$ with
 probability at most $\epsilon_\Xbasis$.
 We remark that this way of finding a lower bound for $\sum_n b_n Q_{\Xbasis,
 \mu_n}$ is more direct than the standard one that separately bounds
 $Y_{\Xbasis,0}$ and $Y_{\Xbasis,1}$~\cite{LCWXZ14,BMFBB16,Chau18,Hayashi07}.

 Putting everything together and by setting $\epsilon_\Xbasis =
 \epsilon_\Zbasis = \epsilon_\Zbasis^\text{e} =
 \epsilon_\Zbasis^{\bar{\text{e}}} = \epsilon_{\bar{\gamma}} =
 \epsilon_\text{sec} / \chi$, we conclude that the secret key rate $R$
 satisfies
\begin{equation}
 R = \sum_{n=1}^k b_n Q_{\Xbasis,\mu_n} - \langle Q_{\Xbasis,\mu} \rangle
 \left\{ \frac{\ln[\chi / \epsilon_\text{sec}]}{2s_\Xbasis} \right\}^{1/2}
 \Width \left( \left\{ \frac{b_n}{p_{\mu_n}} \right\}_{n=1}^k \right) -
 p_\Xbasis^2 \left\{ \langle Q_{\Xbasis,\mu} H_2(E_{\Xbasis,\mu}) \rangle +
 \frac{\langle Q_{\Xbasis,\mu} \rangle}{s_\Xbasis} \left[ 6\log_2
 \frac{\chi}{\epsilon_\text{sec}} + \log_2 \frac{2}{\epsilon_\text{cor}}
 \right] \right\} ,
 \label{E:finite-size_key_rate}
\end{equation}
 where $b_n = b_n(e_p)$ is given by Eq.~\eqref{E:b_n_def}.
 Here $e_p$ equals the R.H.S. of Inequality~\eqref{E:e_p_bound} with
 $e_{\Zbasis,1}$ given by Eq.~\eqref{E:e_Z1_bound_method1},
 \eqref{E:e_Z1_bound_method2}, \eqref{E:e_Z1_bound_method3}
 or~\eqref{E:e_Z1_bound_method4}.
 Note that $\chi = 9 = 4+1+4$ for Methods~\ref{Method:conventional}
 to~\ref{Method:direct} and $\chi = 10$ for Method~\ref{Method:aggressive}.
 (Here the first number $4$ comes from the generalized chain rule for smooth
 entropy in Ref.~\cite{LCWXZ14}, the number $1$ comes from the finite-size
 correction of the raw key in Eq.~(B1) of Ref.~\cite{LCWXZ14}, and the last
 number $4$ comes from $\epsilon_{\bar{\gamma}}$, $\epsilon_\Xbasis$,
 $\epsilon_\Zbasis^\text{e}$ as well as either $\epsilon_\Zbasis$ or
 $\epsilon_\Zbasis^{\bar{\text{e}}}$.
 Moreover, $\chi$ for Method~\ref{Method:aggressive} is larger than the rest
 by 1 because of the extra condition on the statistical fluctuation of a lower
 bound of $Y_{\Zbasis,1} e_{\Zbasis,1}$.)
 Interestingly, unlike the schemes used in Refs.~\cite{LCWXZ14,HN14,Chau18},
 the number $\chi$ in our scheme is independent on the number of photon
 intensities $k$ used.
 This is because we directly tackle the finite sample statistical fluctuations
 of quantities like $Y_{\Bbasis,1}$.
 Note however that even though $\chi$ does not depend on $k$, it does not mean
 that one could use arbitrarily large number of photon intensities as decoys
 (so as to obtain better bounds on quantities like $Y_{\Bbasis,1}$) without
 adversely affecting the key rate for a fixed finite $s_\Xbasis$.
 The reason is that $\Width (\{ a_{1n} / p_{\mu_n} \}_{n=1}^k \})$, $\Width
 (\{ a_{2n} / p_{\mu_n} \}_{n=1}^k \})$ and $\Width (\{ b_n / p_{\mu_n}
 \}_{n=1}^k \})$ diverge as $k\to +\infty$ due to divergence of $a_{1n}$,
 $a_{2n}$ and $b_n$~\cite{Chau18} as well as the decrease in $\min \{
 p_{\mu_n} \}_{n=1}^k$.
 Recall that computing $a_{1n}$, $a_{2n}$ and $b_n$ is numerically stable and
 with minimal lost in precision if $\mu_n - \mu_{n+1} \gtrsim 0.1$ for $n=1,2,
 \dots,k-1$~\cite{Chau18}.
 This means the number of photon intensities $k$ used in practice should be
 $\lesssim 10$.
 
\section{Performance Analysis}
\label{Sec:Performance}
 We study the following quantum channel, which models a commonly used 100~km
 long optical fiber in \QKD experiments, to test the performance of this new
 key rate formula in realistic situation.
 The findings here are generic as the general trend and performance improvement
 are also found in other situations including using the same fiber of
 different lengths as well as other randomly generated quantum channels.
 The yield and error rate of that quantum channel is given by $Q_{\Bbasis,\mu}
 = (1+p_\text{ap}) d_\mu$ and $Q_{\Bbasis,\mu} E_{\Bbasis,\mu} = p_\text{dc} +
 e_\text{mis} [1-\exp(-\eta_\text{ch}\mu)] + p_\text{ap} d_\mu / 2$, where
 $d_\mu = 1-(1-2p_\text{dc}) \exp(-\eta_\text{sys}\mu)$.
 Here we fix after pulse probability $p_\text{ap} = 4\times 10^{-2}$, dark
 count probability $p_\text{dc} = 6\times 10^{-7}$, error rate of the optical
 system $e_\text{mis} = 5\times 10^{-3}$.
 In addition, the transmittance of the system $\eta_\text{sys} = 0.1
 \eta_\text{ch}$, and the transmittance of the fiber is given by
 $\eta_\text{ch} = 10^{-0.2 L /10}$ with $L$ is the length of the fiber in km.
 These parameters are obtained from optical fiber experiment on a 100~km long
 fiber in Ref.~\cite{WLGHZG12}; and have been used in
 Refs.~\cite{LCWXZ14,Chau18} to study the performance of decoy-state \QKD in
 the \FK situation.
 We also follow Refs.~\cite{LCWXZ14,Chau18} by using the following security
 parameters:
 $\epsilon_\text{cor} = \kappa = 10^{-15}$, where $\epsilon_\text{sec} =
 \kappa \ell_\text{final}$ with $\ell_\text{final} \approx R s_\Xbasis /
 (p_\Xbasis^2 \langle Q_{\Xbasis,\mu} \rangle)$ is the length of the final
 key measured in bits.
 Note that $\kappa$ can be interpreted as the secrecy leakage per final secret
 bit.

\begin{table}[t]
 \centering
 \begin{tabular}{||c|c|c|c|c|c|c|c|c|c|c|c|c|c|c|c|c|c|c|c|c||}
  \hline\hline
   & \multicolumn{5}{c|}{$k=3$} & \multicolumn{5}{c|}{$k=4$} &
   \multicolumn{5}{c|}{$k=5$} & \multicolumn{5}{c||}{$k=6$} \\
  \cline{2-21}
  $s_\Xbasis$ &
  $R^\text{O}_{-5}$ & $R^\text{A}_{-5}$ & $R^\text{B}_{-5}$ &
   $R^\text{C}_{-5}$ & $R^\text{D}_{-5}$ &
  $R^\text{O}_{-5}$ & $R^\text{A}_{-5}$ & $R^\text{B}_{-5}$ &
   $R^\text{C}_{-5}$ & $R^\text{D}_{-5}$ &
  $R^\text{O}_{-5}$ & $R^\text{A}_{-5}$ & $R^\text{B}_{-5}$ &
   $R^\text{C}_{-5}$ & $R^\text{D}_{-5}$ &
  $R^\text{O}_{-5}$ & $R^\text{A}_{-5}$ & $R^\text{B}_{-5}$ &
   $R^\text{C}_{-5}$ & $R^\text{D}_{-5}$ \\
  \hline
  $10^5$ &
  $0.052$ & $0.300$ & $0.326$ & $0.470$ & $0.254$ &
  $0.027$ & $0.270$ & $0.291$ & $0.487$ & $0.747$ &
  $0.000$ & $0.152$ & $0.156$ & $0.160$ & $0.142$ &
  $0.000$ & $0.052$ & $0.076$ & $0.076$ & $0.000$ \\
  \hline
  $10^6$ &
  $0.294$ & $0.743$ & $0.789$ & $0.835$ & $0.637$ &
  $0.194$ & $0.727$ & $0.763$ & $0.829$ & $1.41$ &
  $0.100$ & $0.660$ & $0.694$ & $0.516$ & $0.484$ &
  $0.055$ & $0.404$ & $0.434$ & $0.407$ & $0.966$ \\
  \hline
  $10^7$ &
  $0.687$ & $1.18$ & $1.23$ & $1.22$ & $1.04$ &
  $0.573$ & $1.27$ & $1.30$ & $1.27$ & $1.84$ &
  $0.421$ & $1.21$ & $1.20$ & $1.20$ & $1.12$ &
  $0.259$ & $0.949$ & $1.01$ & $0.823$ & $1.57$ \\
  \hline
  $10^8$ &
  $1.11$ & $1.43$ & $1.48$ & $1.45$ & $1.64$ &
  $1.04$ & $1.60$ & $1.63$ & $1.59$ & $2.18$ &
  $0.929$ & $1.66$ & $1.68$ & $1.63$ & $1.81$ &
  $0.624$ & $1.32$ & $1.34$ & $1.33$ & $2.00$ \\
  \hline
  $10^9$ &
  $1.51$ & $1.70$ & $1.75$ & $1.72$ & $2.05$ &
  $1.57$ & $1.91$ & $1.94$ & $1.90$ & $2.40$ &
  $1.46$ & $2.04$ & $2.10$ & $2.06$ & $2.37$ &
  $1.08$ & $1.74$ & $1.75$ & $1.71$ & $2.38$ \\
  \hline
  $10^{10}$ &
  $1.87$ & $1.98$ & $2.02$ & $1.99$ & $2.32$ &
  $1.97$ & $2.20$ & $2.22$ & $2.19$ & $2.58$ &
  $1.94$ & $2.40$ & $2.42$ & $2.40$ & $2.72$ &
  $1.72$ & $2.16$ & $2.18$ & $2.14$ & $2.63$ \\
  \hline
  $10^{11}$ &
  $2.20$ & $2.25$ & $2.29$ & $2.26$ & $2.43$ &
  $2.32$ & $2.46$ & $2.48$ & $2.45$ & $2.81$ &
  $2.46$ & $2.67$ & $2.69$ & $2.69$ & $2.88$ &
  $2.18$ & $2.50$ & $2.52$ & $2.48$ & $2.86$ \\
  \hline\hline
 \end{tabular}
 \caption{Comparison between the state-of-the-art key rate $R^\text{O} \equiv
  R^\text{O}_{-5} \times 10^{-5}$ in Ref.~\cite{Chau18} with the key rates in
  Eq.~\eqref{E:finite-size_key_rate} (or more precisely $R^\text{I}_{-5}
  \equiv \max (0,R^\text{I} \times 10^{-5})$) for the dedicated quantum
  channel used in Refs.~\cite{LCWXZ14,Chau18} via Method~I.
  These rate are optimized using the method stated in the main text.
 \label{T:keyrates}}
\end{table}

 Table~\ref{T:keyrates} compares the optimized key rates for the
 state-of-the-art method reported recently Eq.~(3) of Ref.~\cite{Chau18} with
 Eq.~\eqref{E:finite-size_key_rate} for various $s_\Xbasis$ and $k$.
 (This is the best provably secure key rate obtained before the posting of
 the original proposal using McDiarmid inequality by one of us in
 Ref.~\cite{earlier}.)
 The optimized rates are found by fixing the minimum photon intensity to
 $1\times 10^{-6}$, while maximizing over $p_\Xbasis$ as well as all other
 photon intensities $\mu_n$'s and all the $p_{\mu_n}$'s.
 This optimization is done by Monte Carlo method plus simulated annealing with
 a sample size of at least $10^{10}$ for each data entry in
 Table~\ref{T:keyrates}.
 For Method~\ref{Method:aggressive}, the optimized key rate depends on the
 actual $\Zbasis$-basis measurement results.
 Here we simply fix $n^{(i)}$'s to their expectation values.

 The table clearly shows that using McDiarmid inequality improves the
 optimized key rates in almost all cases.
 It also shows that for any method used, the provably secure key rate
 increases as the raw key length $s_\Xbasis$ increases.
 And they all gradually converge to the same infinite-size key rate.
 Besides, the asymptotic key rate generally increases with $k$.
 These are natural as longer $s_\Xbasis$ implies smaller finite-size
 statistical fluctuation and larger number of decoys $k$ used allows better
 estimation of the bounds of various $Y_{\Bbasis,m}$'s and $Y_{\Bbasis,1}
 e_{\Bbasis,1}$'s.

 Among the four methods introduced here, Method~\ref{Method:conventional}
 almost always gives the least provably secure key rate.
 This implies that it is more effective to estimate a lower bound for
 $Y_{\Zbasis,1}$ via estimating an upper bound for $Y_{\Zbasis,1}
 e_{\Zbasis,1}$ plus a lower bound for $Y_{\Zbasis,1} \bar{e}_{\Zbasis,1}$.
 Method~\ref{Method:Ybare} is slightly better than Method~\ref{Method:direct}
 for large $s_\Xbasis$ (say when $\gtrsim 10^8$, the improvement is about a
 few percent).
 Method~\ref{Method:aggressive} is about 5-15\% or so better than
 Method~\ref{Method:direct} when $10^8 \lesssim s_\Xbasis \lesssim 10^{11}$.
 This is not unexpected for the following reason.
 Although Method~\ref{Method:aggressive} is more aggressive than
 Method~\ref{Method:direct} in estimating the statistical fluctuation of
 $e_{\Zbasis,1}$ and hence the key rate, it requires an additional condition
 for lower-bounding $\langle w\rangle$.
 Thus the value of $\chi$ for Method~\ref{Method:aggressive} is 1 greater than
 that of Method~\ref{Method:direct}.
 As a result, for small raw key length, the improvement in estimating
 $e_{\Zbasis,1}$ for Method~\ref{Method:aggressive} may not be able to
 compensate the need to control the statistical fluctuation of one more
 variable.
 Table~\ref{T:keyrates} also depicts that Method~\ref{Method:aggressive} is
 about $5-20\%$ better than Method~\ref{Method:Ybare} when $10^8 \lesssim
 s_\Xbasis \lesssim 10^{11}$.
 Furthermore, for fixed $s_\Zbasis$ and $\kappa$ and a fixed method to compute
 bound for $e_{\Zbasis,1}$, the provably secure key rate reaches a maximum at
 a finite $k$.
 This is not unexpected because even though the $\chi$ we deduce is
 independent of the number photon intensities $k$ used, $\Width({\mathcal W})$
 diverges as $k\to +\infty$.
 Last but not least, in the case of $k=4$, Method~\ref{Method:aggressive}
 always gives the best key rate.
 We do not have a good answer to this observation.
 It is instructive to study why in future.

\section{Summary And Outlook}
\label{Sec:Summary}
 To summarize, for $s_\Xbasis \approx 10^5 - 10^6$, at least one of the four
 methods reported here could produce a provably secure key rate that is at
 least twice that of the state-of-the-art method.
 And for $s_\Xbasis \approx 10^8$, Method~\ref{Method:aggressive} is at least
 40\% better than the state-of-the-art method.
 These improvements are of great value in practical \QKD because the
 computational and time costs for classical post-processing can be quite high
 when the raw key length $s_\Xbasis$ is long.
 More importantly, the McDiarmid inequality method reported here is effective
 to increase the key rate of real or close to real time on demand generation
 of the secret key --- an application that is possible in near future with the
 advancement of laser technology.
 It is instructive to extend our McDiarmid inequality method to handle the
 case of \FK decoy-state measurement-device-independent \QKD and compare it
 with existing methods in literature, such as the one that uses the Chernoff
 bound~\cite{ZYW16} and its extension specifically for decoys with four
 different intensities~\cite{MZZZZW18}.

 In addition to \QKD, powerful concentration inequalities in statistics such
 as McDiarmid inequality could also be used beyond straightforward statistical
 data analysis.
 One possibility is to use it to construct model independent test for physics
 experiments that involve a large number of parameters but with relatively
 few data points.

\begin{acknowledgments}
 This work is supported by the RGC grants 17304716 and 17302019 of the Hong
 Kong SAR Government.
 We would like to thank K.-B. Luk for his discussion on potential applications
 of McDiarmid inequality in physics.
\end{acknowledgments}

\bibliographystyle{apsrev4-1}

\bibliography{mcdiamid.bib}

\end{document}